\documentclass[11pt]{article}
\usepackage{times}
\usepackage[ruled, vlined]{algorithm2e}
\usepackage{amssymb, amsmath}
\usepackage{color,soul}
\usepackage[text={6.5in,9.0in},centering]{geometry}
\usepackage{bibspacing}
\usepackage{graphicx}
\usepackage{subfig}
\usepackage{hyperref}

\newtheorem{theorem}{Theorem}[section]
\newtheorem{lemma}[theorem]{Lemma}

\newtheorem{definition}{Definition}
\newtheorem{corollary}{Corollary}

\newcommand{\sq}{\hbox{\rlap{$\sqcap$}$\sqcup$}}
\newcommand{\qed}{\hspace*{\fill}\sq}
\newenvironment{proof}{\noindent {\bf Proof.}\ }{\qed\par\vskip 4mm\par}

\begin{document}

\title{Window-Based Greedy Contention Management\\for Transactional Memory\footnote{Gokarna Sharma is recommended for the best student paper award.}
\\
       {\Large (Regular Presentation)}
       }

\author{Gokarna Sharma, Brett Estrade, and Costas Busch \\
   Department of Computer Science\\
   Louisiana State University\\
   Baton Rouge, LA 70803, USA\\
   gokarna@csc.lsu.edu, estrabd@lsu.edu, busch@csc.lsu.edu  \\
   }
\date{}
\maketitle \thispagestyle{empty}

\begin{abstract}
We consider greedy contention managers for transactional memory
for $M\times N$ excution windows of transactions with $M$ threads and $N$ transactions per thread.
Assuming that each transaction conflicts with at most $C$ other transactions inside the window,
a trivial greedy contention manager can schedule them within $CN$ time.
In this paper, we show that there are much better schedules.
We present and analyze two new randomized greedy contention management algorithms.
The first algorithm {\sf Offline-Greedy} produces a schedule of length $O(C+N\log(MN))$
with high probability, and gives competitive ratio $O(\log (MN))$ for $C\leq N \log (MN)$.
The offline algorithm depends on knowing the conflict graph.
The second algorithm {\sf Online-Greedy} produces a schedule of length $O(C \log (MN) + N \log^2(MN))$
with high probability which is only a $O(\log (NM))$ factor worse, but does not require knowledge of the conflict graph.
We also give an adaptive version which achieves similar worst-case performance
and $C$ is determined on the fly under execution.
Our algorithms provide new tradeoffs for greedy transaction scheduling that parameterize window sizes
and transaction conflicts within the window.
\end{abstract}

\centerline{{\bf Keywords}: transactional memory, contention managers, greedy scheduling,
execution window.
}

\section{Introduction}

Multi-core architectures present both an opportunity and challenge for multi-threaded software.
The opportunity is that threads will be available to an unprecedented degree,
and the challenge is that more programmers will be exposed to concurrency related synchronization problems
that until now were of concern only to a selected few.
Writing concurrent programs is difficult because of the complexity of ensuring proper synchronization.
Conventional lock based synchronization suffers from well known limitations,
so researchers considered non-blocking transactions as an alternative.
Software Transactional Memory \cite{Shavit95, Her03, Herlihy93transactionalmemory} systems use lightweight  and composable  in-memory software transactions to address concurrency
in multi-threaded systems ensuring safety all the time \cite{Har03, Har05}.


A contention management strategy is responsible for the STM system as a whole to make progress.
If transaction $T$ discovers it is about to conflict with $T'$,
it has two choices, it can pause, giving $T'$ a chance to finish,
or it can proceed, forcing $T'$ to abort.
To solve this problem efficiently,
$T$ will consult the contention manager module which choice to make.
Of particular interest are {\em greedy contention managers}
where a transaction starts again immediately after every abort.
Several (greedy) contention managers have been proposed in the literature.
However, most contention managers have been assessed only experimentally by specific benchmarks.
There is a small amount of work in the literature which analyzes formally
the performance of contention managers.
The competitive ratio results are not encouraging since the bounds are not tight.
For example with respect to the $O(s)$ bound in \cite{Att06},
when the number of resources increases, the performance degrades linearly.
A question arises whether someone can achieve tighter bounds.
A difficulty in obtaining tight bounds is that the algorithms studied in
\cite{Att06,Gue05a,Gue05b, scherer05, Schneider09}
apply to the {\em one-shot scheduling problem}, where each thread issues a single transaction.
One-shot problems can be related to graph coloring.
It can be shown that the problem of finding the chromatic number of a graph
can be reduced to finding an optimal schedule for a one-shot problem.
Since it is known that graph coloring is a very hard problem to approximate,
the one-shot problem is very hard to approximate too \cite{Schneider09}.

In order to obtain better formal bounds,
we propose to investigate execution window of transactions (see the left part of Figure \ref{figure:window}),
which has the potential to overcome the limitations of coloring in certain circumstances.
An $M \times N$ window of transactions $W$
consists of $M$ threads with an execution sequence of $N$ different transactions per thread.
Let $C$ denote the maximum number of conflicting transactions for any transaction in the window
($C$ is the maximum degree of the respective conflict graph of the window).
A straightforward upper bound is $\min(C N, MN)$,
since $CN$ follows from the observation that
each transaction in a thread may be delayed at most $C$ time steps by its conflicting transactions,
and $MN$ follows from the serialization of the transactions.
If we partition the window into $N$ one-shot transaction sets,
each of size $M$, then the competitive ratio using the one-shot analysis results is $O(s N)$.
When we use the Algorithm {\sf RandomizedRounds} \cite{Schneider09} $N$ times
then the completion time is in the worst case $O(C N \log n)$ (for some appropriate choice of $n$).

\begin{figure*}[ht]
\centerline{\subfloat[Before execution]{\includegraphics[height=1.6in,width=2.3in]{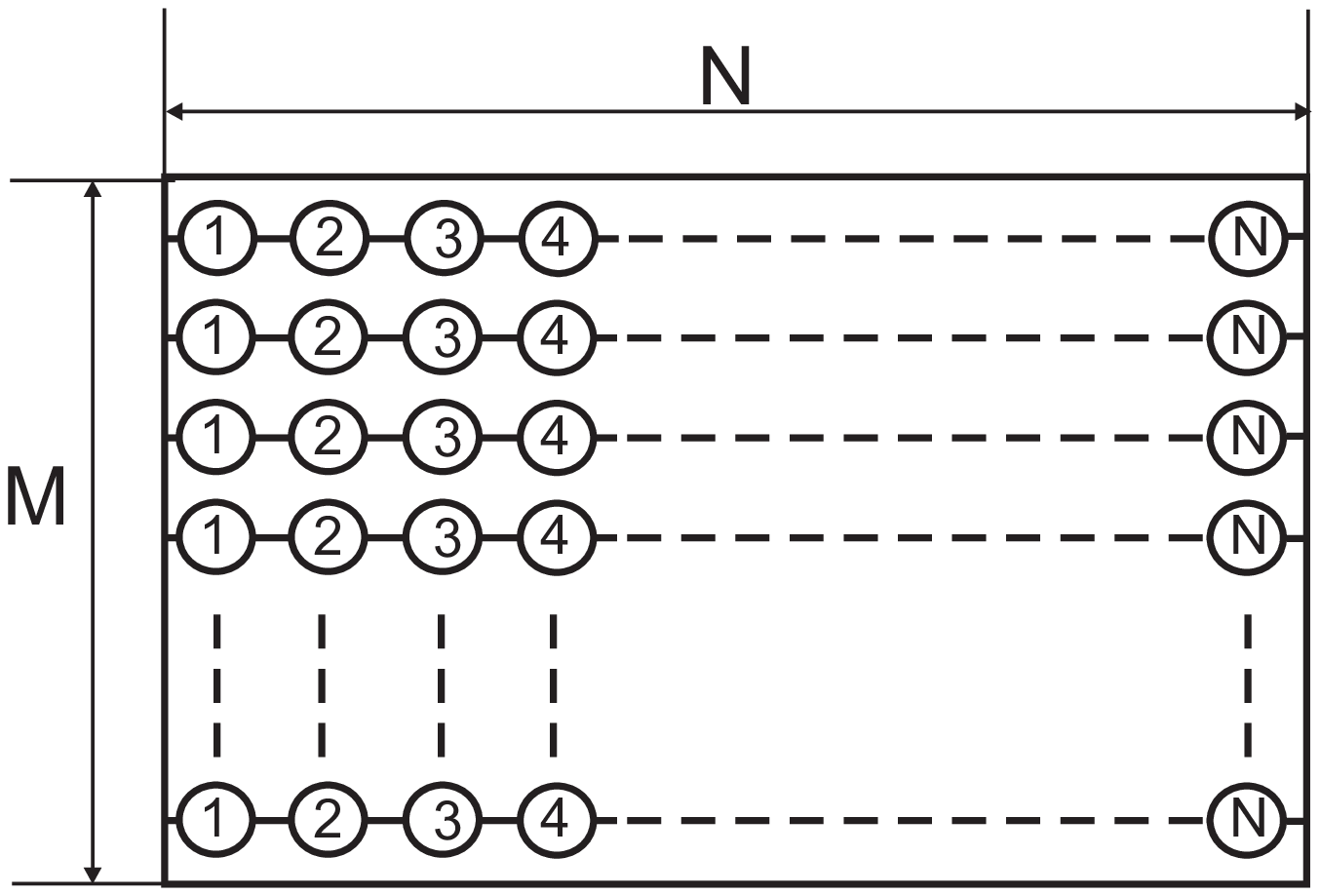}}
\hfil
\subfloat[After execution]{\includegraphics[height=1.9in,width=3.3in]{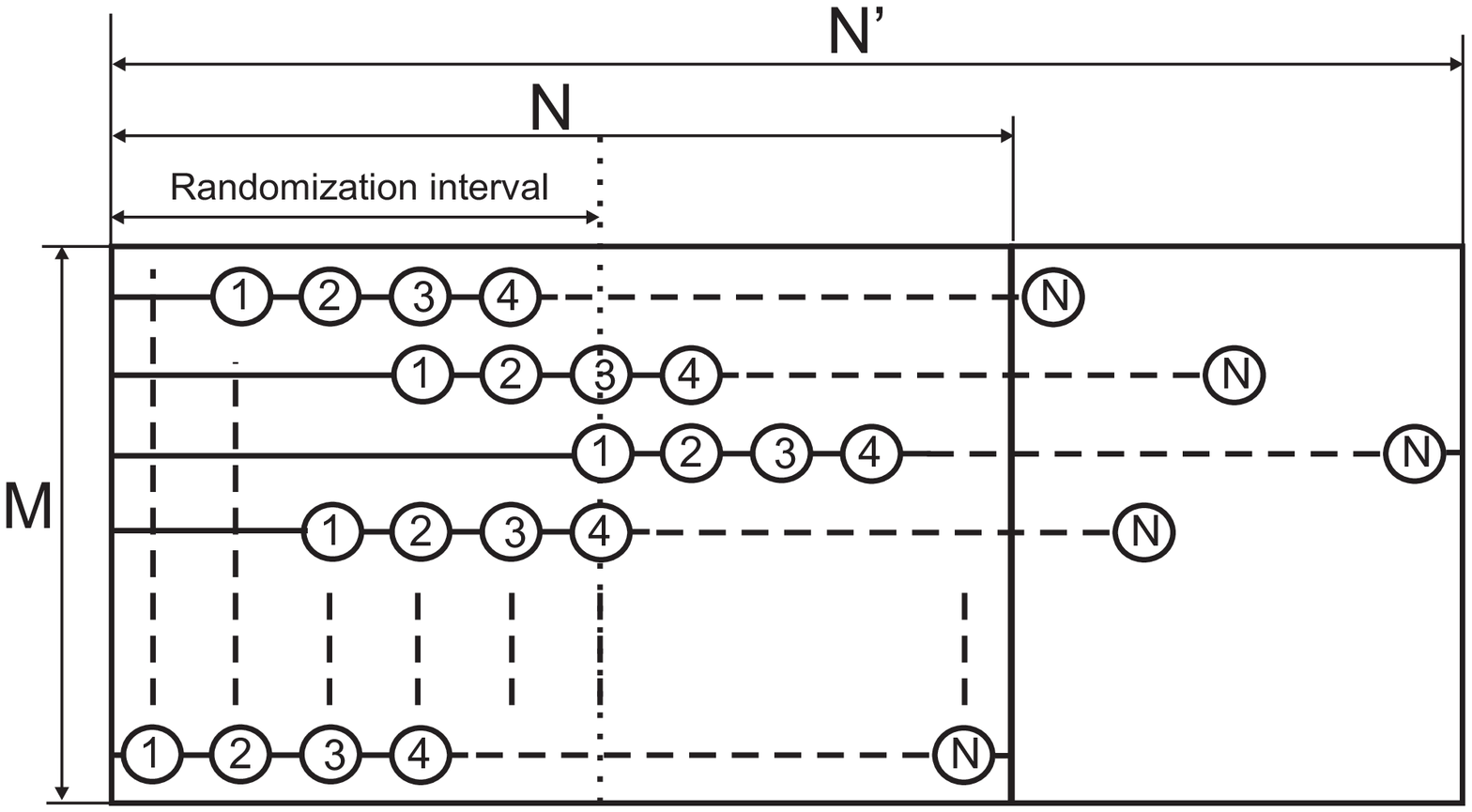}}
}
\caption{Execution window model for transactional memory}
\label{figure:window}
\end{figure*}

We have results that indicate
that we can obtain better bounds under certain circumstances in the window.
We present two randomized greedy algorithms
transactions are assigned priorities values,
such that for some random initial interval in the beginning of the window $W$
each transaction is in low priority mode and then after the random period expires
the transactions switch to high priority mode.
In high priority mode the transaction can only be aborted by other
high priority transactions.
The random initial delays have the property that the conflicting transactions
are shifted inside their window and their execution times may not coincide
(see the right part of Figure \ref{figure:window}).
The benefit is that conflicting transactions can execute at different time slots 
and potentially many conflicts are avoided.
The benefits become more apparent in scenarios where 
the conflicts are more frequent inside the same column transactions
and less frequent between different column transactions.

\paragraph{Contributions:}
We propose the contention measure $C$ within the window 
to allow more precise statements about the worst-case complexity bound of any contention management algorithm.
We give two window-based randomized greedy algorithms 
for the contention management in any execution window $W$.
Our first Algorithm {\sf Offline-Greedy} gives a schedule of length $O(C + N \log(MN))$ with high probability,
and improves on one-shot contention managers from a worst-case perspective.
The algorithm is offline in the sense that it uses explicitly the conflict graph 
of the transactions to resolve the conflicts. 
Our second Algorithm {\sf Online-Greedy} produces a schedule of length $O(C \log (MN)+ N \log^2(MN))$ 
with high probability, which is only a factor of $O(\log (MN))$ worse in comparison to  {\sf Offline-Greedy}.
The benefit of the online algorithm is that does not need to know the conflict graph of the transactions
to resolve the conflicts.
The online algorithm uses as a subroutine Algorithm {\sf RandomizedRounds} \cite{Schneider09}.
We also give a third algorithm {\sf Adaptive-Greedy} which is the adaptive version
of the previous algorithms which achieves similar worst-case performance
and adaptively guesses the value of the contention measure $C$. 

The technique we use for the analysis of these algorithms 
is similar to the one used by Leighton {\it et al.}~\cite{LMR94}
to analyze an online packet scheduling problem.
Moreover, one advantage of our algorithms is that if the conflicts in the window are bounded by $C \leq N \log MN$
then the upper bounds we have obtained is within poly-logarithmic factors from optimal,
since $N$ is a lower bound for the execution time. 
By finding window sizes in the program execution where $C$ is small compared to $N$ our algorithm
provide better bounds than previously known algorithms. 

We prove the existence of an algorithm based on dynamic programming to 
find in polynomial time the optimal decomposition for any arbitrary window $W$,
into sub-windows $W_1, \ldots, W_k$,
such the maximum contention {\em density} in each is the smallest possible.
The density simply measures how much larger is $C$ with respect to the number of transactions per thread.
By applying our greedy contention management algorithms in the sub-windows
we can obtain schedules which are asymptotically better than 
executing the algorithm in the whole window $W$.   

\paragraph{Outline of Paper:}
The rest of the paper is organized as follows: In Section \ref{section:related}, we discuss the related work.
We present the transactional memory model in Section \ref{section:model}. 
We present and formally analyze an offline randomized greedy algorithm in Section \ref{section:offlinealgorithm}. 
The online version is given in Section \ref{section:onlinealgorithm}. 
In Section \ref{section:adaptivealgorithm}, we describe the adaptive version of the aforementioned algorithms. 
We discuss the issues of window decomposition for the optimal window generation 
in Section \ref{section:decomposition}. Section \ref{section:conclusion} concludes the paper.

\section{Related Work}
\label{section:related}
Transactional Memory (TM) has been proposed in the early nineties as an alternative implementation of mutual exclusion
that avoids many of the drawbacks of locks 
(e.g., deadlock, reliance on the programmer to associate shared data with locks,
priority inversion, and failures of threads while holding locks) \cite{Herlihy93transactionalmemory}. 
A few years later the term Software Transactional Memory (STM) was suggested by Shavit and Touitou \cite{Shavit95}
and a so called Dynamic STM (DSTM) for dynamic data structures which uses a contention manager
as an independent module was proposed \cite{Her03}. 
DSTM is a practical obstruction-free STM system that seeks advice from the contention 
manager module to either wait or abort an transaction at the time of conflict.

Several contention managers have been proposed in the literature. 
Most of them have been assessed by specific benchmarks only and not analytically.
A comparison of contention managers based on different benchmarks can be found in \cite{Sch05, Sch04, Ramadan08, scherer05}.
They found out that the choice of the contention manager varies with the complexity of the considered benchmark.
The more detailed analysis of the performance of different contention managers
in complex benchmarks has recently been studied by Ansari {\it et al.}~\cite{Ansari09}.
From all the aforementioned references, 
it has been turned out that the coordination cost and the overhead involved in contention management is very high.

The first formal analysis of the performance of a contention manager was given by Guerraoui {\it et al.}~\cite{Gue05a}
which presented the {\sf Greedy} contention manager and proved that it achieves $O(s^{2})$
competitive ratio in comparison to the optimal off-line schedulers for $n$ concurrent transactions that share $s$ objects.
Later, Guerraoui {\it et al.}~\cite{Gue05b} studied the impact of transaction
failures on contention management and proved the $O(ks^{2})$
competitive ratio when some running transaction may abort $k$ times and then eventually commits.
Attiya {\it et al.}~\cite{Att06} improved the result of \cite{Gue05a} to $O(s)$,
and the result of \cite{Gue05b} to  $O(ks)$,
which are significant improvements over the competitive ratio of {\sf Greedy}. The also proved the matching lower bound of $\Omega(s)$ for the competitive ratio for deterministic work-conserving algorithms which schedule as many transactions as possible.

The complexity measures provided by the aforementioned 
studies are not satisfying as they are based on number of shared resources only. 
One can notice that number of shared resources in total is not really related 
to the actual conflicting transactions potentially encountered by an transaction.
Recently, Schneider and Wattenhofer \cite{Schneider09}
analyzed some of the issues related to the number of potential conflicts; and
presented a deterministic algorithm {\sf CommitBounds}
with competitive ratio $\Theta(s)$ for $n$ concurrent transactions using $s$ shared resources
and a randomized algorithm {\sf RandomizedRounds} with makespan
$O(C \log n)$, for the one-shot problem of a set of $M$ transactions in separate threads 
with $C$ conflicts (assuming unit delays for transactions),
with high probability (proportional to $1-n^{-1}$). 
Which means, {\sf RandomizedRounds} is only a factor of $\log n$ from optimal, with high probability,
for the case where $C < M$. 
However, if other transactions comes into play that are able to reduce 
the parallelism by a factor of $k$, 
the approximation of {\sf RandomizedRounds} also worsens by a factor of $k$.\
While previous studies showed that contention managers {\sf Polka} \cite{Sch05} and {\sf SizeMatters} \cite{Ramadan08} 
exhibits good overall performance for variety of benchmarks, 
this work showed that they may perform exponentially worse than {\sf RandomizedRounds} from a worst-case perspective.
\section{Execution Window Model} \label{section:model}

We consider a model that is based on a $M \times N$ execution window $W$
consisting of a set of transactions $W = \{(T_{11},\cdots, T_{1N}),$
$(T_{21},\cdots, T_{2N}),$
$\ldots,(T_{M1},\cdots,T_{MN})\}$
executed by the $M$ threads running on $M$ processors $P_1, \cdots, P_M$ 
where each thread issues $N$ transactions
in a sequence. 
For the simplicity of the analysis we assume that a single processor runs one thread only, i.e., 
in total at most $M$ threads are running concurrently. 
A thread running on processor $P_i$ executes transactions $T_{i1}, \cdots, T_{iN}$
one after the other and transaction $T_{ij}$ is executed as soon as $T_{i(j-1)}$ has completed or committed. 

Transactions share a set of objects $\Psi =\{O_1,\cdots,O_s\}$. 
Each transaction $T_{ij}$ may use at most $s$ different objects.
Each transaction is a sequence of actions that is either a 
read to some shared resource $O_l$, a write to some shared resource $O_k$, a commit, or an abort.
Concurrent write-write actions or read-write actions to shared objects
by two or more transactions cause conflicts between transactions.
Each transaction completes with a commit when each action performed without conflicts.
If conflicts occur then a transaction either aborts,
or it may commit and force to abort all other conflicting transactions.
In a {\em greedy} schedule,
if a transaction aborts then it immediately attempts to execute again until it commits.

Each transaction $T_{ij}$ has execution time duration $\tau_{ij}$
which is greater than 0.
Here, for simplicity, 
we assume that $\tau_{ij} = 1$, i.e., each transaction needs one time unit to execute. 
We also assume that the execution of the transactions starts at time 0 
and the execution time advances synchronously for all threads step by step.
We also assume that all transactions inside the execution window are correct,
i.e., there are no faulty transactions.
Our results can be extended by relaxing these assumptions.

The {\em makespan} of a schedule for a set of transactions $\Gamma$ 
is defined as the duration from the start of the schedule, i.e., 
the time when some transaction $T_{ij}\in \Gamma$ is available for scheduling, 
until all transactions in $\Gamma$ have committed.
The makespan of the transaction scheduling algorithm for the sequences of transactions 
can be compared to the makespan of an optimal off-line scheduling algorithm, 
which is denoted by OPT. 
We evaluate the efficiency of our new contention management algorithms
by comparing their makespan with the makespan of the optimal off-line scheduler. 

\begin{definition}[Competitive Ratio] 
The competitive ratio of the combination of $(A, \Gamma)$ 
for a contention management algorithm $A$ under a set of jobs $\Gamma$ 
is defined as $$CR(A, \Gamma) = \frac{makespan(A, \Gamma)}{makespan(\textsc{OPT}, \Gamma)}.$$
\end{definition}


\paragraph{Conflict Graph:} 
For a set of transactions $V \subseteq \Gamma$, 
we use the notion of conflict graph $G=(V,E)$.
The neighbors of a transaction $T$ in the conflict graph are denoted by $N_{T}$ 
and represent all transactions that have a conflict with $T$ in $G$. 
The degree $d_{T}$ of $T$ in the graph corresponds 
to the number of its neighbors in the conflict graph, i.e., $d_{T}=|N_{T}|$.
Note $d_{T} \leq |V|$.
The congestion $C$ of the window $W$ is the largest degree of the conflict graph $G' = (W,E')$,
which consists of all the transactions in the window. 

\section{Offline Algorithm} \label{section:offlinealgorithm}

We present Algorithm {\sf Offline-Greedy} (Algorithm \ref{algorithm:offline-greedy})
which is an offline greedy contention resolution algorithm
that uses the conflict graph explicitly to resolve
conflicts of transactions.
First, we divide the time into frames of duration $\Phi = \Theta(\ln(MM))$.
Then, each thread $P_i$ is assigned an initial time period consisting of $R_i$ frames
(with total duration $R_i \cdot \Phi$),
where $R_i$  is chosen randomly, independently
and uniformly, from the range $\left[0, \alpha -1\right]$,
where $\alpha = C / \ln(MN)$.
Each transaction has two priorities: $low$ or $high$ associated with them.
Transaction $T_{ij}$ is initially in low priority.
Transaction $T_{ij}$ switches
to high priority (or normal priority) in the first time step of frame $F_{ij} = R_i + (j-1)$
and remains in high priority thereafter until it commits.
The priorities are used to resolve conflicts.
A high priority transaction may only be aborted by another high priority transaction.
A low priority transaction is always aborted if it conflicts with a high priority transaction.

Let $G_{t}$ denote the conflict graph of transactions at time $t$
where each transaction corresponds to a node and two transactions are connected
with an edge if they conflict in at least one shared resource.
Note that the maximum degree of $G_{t}$ is bounded by $C$
for the transactions in window $W$.
At each time step $t$ we select to commit a maximal independent set of transactions in $G_{t}$.
We first select a maximal independent set $I_H$ of high priority transactions
then remove this set and its neighbors from $G_{t}$,
and then select a maximal independent set $I_L$
of low priority transactions from the remaining conflict graph.
The transactions that commit are $I_H \cup I_L$.

The intuition behind the algorithm is as follows:
Consider a thread $i$ and its first transaction in the window $T_{i1}$.
According to the algorithm, $T_{i1}$ becomes high priority in the beginning
of frame $F_{i1}$.
Because $R_i$ is chosen at random among $\alpha C/\ln(MN)$ positions
it is expected that $T_{i1}$ will conflict with at most $O(\ln(MN))$
transactions which become simultaneously
high priority in the same time frame (in $F_{ij}$).
Since the duration of a time frame is $\Phi = \Theta(\ln(MN))$,
transaction $T_{i1}$ and all its high priority conflicting transactions will
be able to commit by the end of time frame $Y_i$,
using the conflict resolution graph.
The initial randomization period of $R_i \cdot \Phi$ frames
will have the same effect to the
remaining transactions of the thread $i$,
which will also commit within their chosen high priority frames.

\begin{algorithm}[t]
{\small
\KwIn{A $M \times N$ window $W$ of transactions with $M$ threads each with $N$ transactions,
where $C$ is the maximum number of transactions that a transaction can conflict within the window\;}
\KwOut{A greedy execution schedule for the window of transactions $W$\;}
\BlankLine
Divide time into time frames of duration $\Phi = 1 + (e^2 + 2) \ln(MN)$\;
Each thread $P_i$ chooses a random number $R_i \in [0, \alpha -1]$ for $\alpha = C / \ln(MN)$\;
\ForEach{time step $t = 0,1,2,\ldots$}{
\textbf{Phase 1: Priority Assignment}\;
\ForEach{transaction $T_{ij}$}{
$F_{ij} \leftarrow R_i + (j-1)$\;
\eIf{$t < F_{ij} \cdot \Phi$}
{
$Priority(T_{ij}) \gets Low$\;
}{
$Priority(T_{ij}) \gets High$\;
}}
\textbf{Phase 2: Conflict Resolution}\;
\Begin{
Let $G_{t}$ be the conflict graph at time $t$\;
Compute $G_{t}^H$ and $G_{t}^L$, the subgraphs of $G_{t}$ induced by high and low priority nodes, respectively\;
Compute $I_H \gets I(G_{t}^H)$, maximal independent set of nodes in graph $G_{t}^H$\;
$Q \gets $ low priority nodes adjacent to nodes in $I_H$\;
Compute $I_L = I(G_{t}^L - Q)$,
maximal independent set of nodes in graph $G_{t}^L$ after removing $Q$ nodes\;
Commit $I_H \cup I_L$\;
}
}
\caption{{\sf Offline-Greedy}}
\label{algorithm:offline-greedy}
}
\end{algorithm}

\subsection{Analysis of Offline Algorithm}
We study two classic efficiency measures for the analysis of our contention management algorithm:
(a) the makespan, which gives the total time to complete all the $MN$ transactions
in the window; and
(b) the response time of the system, which gives how much time a transaction takes to commit.

According to the algorithm, when a transaction $T_{ij}$ enters into the system,
it will be in low priority until $F_{ij}$ starts. As soon as $F_{ij}$ starts,
it will enter into its respective frame and begin executing in high priority.
Let $A$ denote the set of conflicting transactions with $T_{ij}$.
Let $A' \subseteq A$ denote the subset of conflicting transactions of $T_{ij}$
which become high priority during frame $F_{ij}$ (simultaneously with $T_{ij}$).

\begin{lemma}
\label{lemma:commit-frame}
If $|A'| \leq \Phi - 1$ then transaction $T_{ij}$ will commit in frame $F_{ij}$.
\end{lemma}

\begin{proof}
Due to the use of the high priority independent sets in the conflict graph $G_t$,
if in time $t$ during frame $F_{ij}$
transaction $T_{ij}$ does not commit, then some
conflicting transaction in $A'$ must commit.
Since there are at most $\Phi - 1$ high priority conflicting transactions,
and the length of the frame $F_{ij}$ is at most $\Phi$,
$T_{ij}$ will commit by the end of frame $F_{ij}$.
\end{proof}

We show next that it is unlikely that $|A'| > \Phi - 1$.
We will use the following version of the Chernoff bound:

\begin{lemma}[Chernoff bound 1]
\label{lemma:chernoff1}
Let $X_1, X_2, \ldots, X_n$ be independent Poisson trials
such that, for $1 \leq i \leq n$, ${\bf Pr}(X_i = 1) = pr_i$,
where $0 < pr_i < 1$.
Then, for $X = \sum_{i=1}^{n} X_i$,
$\mu = {\bf E}[X] = \sum_{i=1}^{n} pr_i$,
and any $\delta > e^2$,
${\bf Pr}(X > \delta \mu) < e^{-\delta \mu}.$
\end{lemma}

\begin{lemma}
\label{lemma:one-transaction}
$|A'| > \Phi - 1$ with probability at most $(1/MN)^2$.
\end{lemma}

\begin{proof}
Let $A_k \subseteq A$, where $1 \leq k \leq M$,
denote the set of transactions of thread $P_k$ that conflict
with transaction $T_ij$.
We partition the threads $P_1, \ldots, P_M$
into $3$ classes
$Q_0$, $Q_1$, and $Q_3$,
such that:
\begin{itemize}
\item
$Q_0$ contains every thread $P_k$ which either $|A_k| = 0$,
or $|A_k| > 0$ but the positions of the transactions in $A_k$
are such that it is impossible to overlap with $F_{ij}$
for any random intervals $R_i$ and $R_k$.
\item
$Q_1$ contains every thread $P_k$ with $0 < |A_k| < \alpha$,
and at least one of the transactions in $A_k$ is positioned
so that it is possible to overlap with with frame $F_{ij}$
for some choices of the random intervals $R_i$ and $R_k$.
\item
$Q_2$ contains every thread $P_k$ with $\alpha \leq |A_k|$.
Note that $|Q_2| \leq C/\alpha = \ln(NM)$.
\end{itemize}

Let $Y_k$ be a random binary variable,
such that $Y_k = 1$ if in thread $P_k$ any of the transactions in $A_k$
becomes high priority in $F_{ij}$ (same frame with $T_{ij}$),
and $Y_k = 0$ otherwise.
Let $Y = \sum_{k=1}^{M} Y_k$.
Note that $|A'| = Y$.
Denote $pr_k = {\bf Pr}(Y_k = 1)$.
We can write $Y = Z_0 + Z_1 + Z_2$,
where $Z_\ell = \sum_{P_k \in Q_\ell} Y_k$, where $0 \leq \ell \leq 2$.
Clearly, $Z_0 = 0$.
and $Z_2 \leq |Q_2| \leq \ln(MN)$.

Recall that for each thread $P_k$ there is a random initial interval
with $R_k$ frames,
where $R_k$ is chosen uniformly at random in $[0,\alpha-1]$.
Therefore, for each $P_k \in Q_1$,
$0 < pr_k \leq |A_k| / \alpha < 1$,
since there are $|A_k| < \alpha$ conflicting transactions
in $A_i$ and there are at least $\alpha$ random choices
for the relative position of transaction $T_{ij}$.
Consequently,
$$\mu = {\bf E}[Z_1] = \sum_{P_k \in Z_1} pr_k
\leq \sum_{P_k \in Z_1} \frac {|A_k|} {\alpha}
= \frac{1} {\alpha} \cdot \sum_{P_k \in Z_1} |A_k|
\leq \frac C \alpha
\leq {\ln(MN)}.$$
By applying the Chernoff bound of Lemma \ref{lemma:chernoff1}
we obtain that
$${\bf Pr}(Z_1 > (e^2+1) \mu) < e^{-(e^2+1) \mu} < e^{-2 \ln(MN)} = (MN)^{-2}.$$
Since $Y = Z_0 + Z_1 + Z_2$, and $Z_2 \leq \ln(MN)$,
we obtain ${\bf Pr}(|A'| = Y > (e^2 + 2) \mu = \Phi - 1) < (MN)^{-2}$, as needed.
\end{proof}

\begin{theorem}[makespan of {\sf Offline-Greedy}]
\label{theorem:offline-makespan}
Algorithm {\sf Offline-Greedy} produces a schedule of length $O(C + N \log(MN))$
with probability at least $1-\frac{1}{MN}$.
\end{theorem}

\begin{proof}
From Lemmas \ref{lemma:commit-frame} and \ref{lemma:one-transaction}
the frame length $\Phi$ does not suffice to
commit transaction $T_{ij}$ within frame $F_{ij}$ (bad event)
with probability at most ${NM}^{-2}$.
Considering all the $MN$
transactions in the window a bad event occurs
with probability at most $MN \cdot MN^{-2} = MN^{-1}$.
Thus, with probability at least $1 - MN^{-1}$ all transactions
will commit with the frames that they become high priority.
The total time used by any thread is bounded by $(\alpha + N)\cdot \Phi = O(C + N \log(MN))$.
\end{proof}

Since $N$ is a lower bound for the makespan,
Theorem \ref{theorem:offline-makespan} implies
the following competitive ratio for the $M \times N$ window $W$:

\begin{corollary}[competitive ratio of {\sf Offline-Greedy}]
When $ C \leq N \cdot \ln(MN)$,
$CR(\mbox{{\sf Offline-Greedy}}, W) = O(\log(NM))$,
with high probability.
\end{corollary}

The following corollary follows immediately from
Lemmas \ref{lemma:commit-frame} and \ref{lemma:one-transaction}:

\begin{corollary}[response time of {\sf Offline-Greedy}]
\label{theorem:offline-response}
The time that a transaction $T_{ij}$ needs to commit from the moment it starts
is $O(C + j \cdot\log(MN))$ with probability at least $1-\frac{1}{(MN)^2}$.
\end{corollary}

\section{Online Algorithm}
\label{section:onlinealgorithm}

We present Algorithm {\sf Online-Greedy} (Algorithm \ref{algorithm:online-greedy}),
which is online in the sense that it does not depend
on knowing the dependency graph to resolve conflicts.
This algorithm is similar to Algorithm \ref{algorithm:offline-greedy}
with the difference that in the conflict resolution phase
we use as a subroutine a variation of Algorithm {\sf RandomizedRounds}
proposed by Schneider and Wattenhofer \cite{Schneider09}.
The makespan of the online algorithm is slightly worse than the offline algorithm,
since the duration of the phase is now $\Phi' = O(\log^2(MN))$.

There are two different priorities associated with each transaction under this algorithm.
The pair of priorities for a transaction is given as a vector
$\langle \pi^{(2)}, \pi^{(1)} \rangle$,
where $\pi^{(2)}$ represents the Boolean priority value $low$ or $high$
(with respective values 1 and 0)
as described in Algorithm \ref{algorithm:offline-greedy},
and $\pi^{(1)} \in [1, M]$ represents the random priorities used in
Algorithm {\sf RandomizedRounds}.
The conflicts are resolved in lexicographical order based on the priority vectors,
so that vectors with lower lexicographic order have higher priority.

When a transaction $T$ enters the system,
it starts to execute immediately in low priority ($\pi^{(2)} = 1$)
until the respective randomly chosen time frame $F$ starts
where it switches to high priority ($\pi^{(2)} = 0$).
Once in high priority, the field $\pi^{(1)}$ will be used to resolve conflicts
with other high priority transactions.
A transaction chooses a discrete number $\pi^{(1)}$
uniformly at random in the interval $[1, M]$ on start of the frame $F_{ij}$,
and after every abort.
In case of a conflict with another high priority transaction $K$ but
which has higher random number ($\pi^{(1)}$) than $T$,
then $T$ proceeds and $K$ aborts.
The procedure $Abort(T,K)$ aborts transaction $K$
and $K$ must hold off on restarting (i.e. hold off attempting to commit)
until $T$ has been committed or aborted.

\begin{algorithm}[t]
{\small
\KwIn{A $M \times N$ window $W$ of transactions with $M$ threads each with $N$ transactions,
where $C$ is the maximum number of transactions that a transaction can conflict within the window\;}
\KwOut{A greedy execution schedule for the window of transactions $W$\;}
\BlankLine
Divide time into time frames of duration
$\Phi' = 16 e \Phi \ln (MN)$\;
Associate pair of priorities $\langle \pi_{ij}^{(2)}, \pi_{ij}^{(1)} \rangle$
to each transaction $T_{ij}$\;
Each thread $P_i$ chooses a random number $R_i \in [0, \alpha -1]$
for $\alpha = C / \ln(NM)$\;
\ForEach{time step $t = 0,1,2,\ldots$}{
\textbf{Phase 1: Priority Assignment}\;
\ForEach{transaction $T_{ij}$}{
$F_{ij} \leftarrow R_i + (j-1)$\;
\eIf{$t < F_{ij} \cdot \Phi'$}
{
Priority $\pi_{ij}^{(2)} \gets 1 ~(Low)$\;
}{
Priority $\pi_{ij}^{(2)} \gets 0 ~(High)$\;
}}
\textbf{Phase 2: Conflict Resolution}\;
\Begin{
\If{$\pi_{ij}^{(2)} == 0$ ($T_{ij}$ has high priority)}{
\textbf{On (re)start} of transaction $T_{ij}$\;
\Begin{
$\pi_{ij}^{(1)} \leftarrow $ random integer in $[1,M]$\;
}
\BlankLine
\textbf{On conflict} of transaction $T_{ij}$ with high priority transaction $T_{kl}$\;
\Begin{
\eIf{$\pi_{ij}^{(1)} < \pi_{kl}^{(1)}$}
{
$Abort(T_{ij}, T_{kl})$\;
}{
$Abort(T_{kl}, T_{ij})$\;
}
}
}
}
}
\caption{{\sf Online-Greedy}}
\label{algorithm:online-greedy}
}
\end{algorithm}

\subsection{Analysis of Online Algorithm}

In the analysis given below,
we study the makespan and the response time of Algorithm {\sf Online-Greedy}.
The analysis is based on the following adaptation of the response time analysis
of a one-shot transaction problem with Algorithm {\sf RandomizedRounds} \cite{Schneider09}.
It uses the following Chernoff bound:

\begin{lemma}[Chernoff bound 2]
\label{lemma:chernoff2}
Let $X_1, X_2, \ldots, X_n$ be independent Poisson trials
such that, for $1 \leq i \leq n$, ${\bf Pr}(X_i = 1) = pr_i$,
where $0 < pr_i < 1$.
Then, for $X = \sum_{i=1}^{n} X_i$,
$\mu = {\bf E}[X] = \sum_{i=1}^{n} pr_i$,
and any $0 < \delta \leq 1$,
${\bf Pr}(X < (1-\delta) \mu) < e^{-\delta^2 \mu / 2}.$
\end{lemma}

\begin{lemma}
\label{theorem:roger}
\textbf{(Adaptation from Schneider and Wattenhofer \cite{Schneider09})}
Given a one-shot transaction scheduling problem with $M$ transactions,
the time span a transaction $T$ needs from its first start
until commit is $16 e (d_{T}+1) \log n$
with probability at least $1-\frac{1}{n^2}$,
where $d_{T}$ is the number of transactions conflicting with $T$.
\end{lemma}

\begin{proof}
Consider the conflict graph $G$.
Let $N_T$ denote the set of conflicting transactions for $T$
(these are the neighbors of $T$ in $G$).
We have $d_{T} = |N_T| \leq m$.
Let $y_T$ denote the random priority number choice of $T$ in range $[1,M]$.
The probability that for transaction $T$
no transaction $K \in N_{T}$ has the same random number is:
$${\bf Pr}(\nexists K \in N_{T} | y_T = y_K)= \left (1-\frac{1}{M} \right )^{d_{T}}
\geq \left ( 1-\frac{1}{M} \right )^M \geq \frac{1}{e}.$$
The probability that $y_{T}$ is at least as small as $y_K$
for any transaction $K \in N_{T}$ is $\frac{1}{d_{T}+1}$.
Thus, the chance that $y_{T}$ is smallest and different among all its neighbors in $N_T$
is at least $\frac{1}{e (d_{T}+1)}$.
If we conduct $16 e (d_{T}+1) \ln n$ trials,
each having success probability $\frac{1}{e (d_{T}+1)}$,
then the probability that the number of successes $Z$ is less than $8 \ln n$ becomes:
${\bf Pr}(Z < 8 \cdot \ln n) <e^{-2\cdot \ln n} =\frac{1}{n^2}$,
using the Chernoff bound of Lemma \ref{lemma:chernoff2}.
\end{proof}

\begin{theorem} [makesspan of {\sf Online-Greedy}]
\label{theorem:online-makespan}
Algorithm {\sf Online-Greedy} produces a schedule of length $O(C \log (MN) + N \log^2(MN))$
with probability at least $1-\frac{2}{MN}$.
\end{theorem}

\begin{proof}
According to the algorithm, a transaction $T_{ij}$ becomes high priority ($\pi_{ij}^{(1)} = 0$)
in frame $F_{ij}$.
When this occurs the transaction will start to compete with other transactions which became
high priority during the same frame.
Lemma \ref{lemma:commit-frame} from the analysis of Algorithm \ref{algorithm:offline-greedy},
implies that the effective degree of $T_{ij}$ with respect to high priority transactions
is $d_T > \Phi - 1$ with probability at most $(MN)^{-2}$ (we call this bad event $A$).
From Lemma \ref{theorem:roger},
if $d_T \leq \Phi - 1$,
the transaction will not commit within $16 e (d_{T}+1) \log n \leq \Phi'$ time slots
with probability at most $(MN)^{-2}$ (we call this bad event $B$).
Therefore, the bad event that $T_{ij}$ does not commit in $F_{ij}$ occurs
when either bad event $A$ or bad event $B$ occurs,
which happens with probability at most $(MN)^{-2} + (MN)^{-2} = 2(MN)^{-2}$.
Considering now all the $MN$ transactions,
the probability of failure is at most $2/MN$.
Thus, with probability at least $1- 2/NM$, every transaction $T_{ij}$
commits during the $F_{ij}$ frame.
The total duration of the schedule is bounded by
$(\alpha + N)\Phi' = O(C \log (MN)+ N \log^2(MN)).$
\end{proof}

\begin{corollary}[competitive ratio of {\sf Online-Greedy}]
When $ C \leq N \cdot \ln(MN)$,
$CR(\mbox{{\sf Online-Greedy}}, W) = O(\log^2(NM))$,
with high probability.
\end{corollary}

\begin{corollary}[response time of {\sf Online-Greedy}]
\label{theorem:online-timespan}
The time that a transaction $T_{ij}$ needs to commit from the moment it starts
is $O(C \log (MN) + j \cdot \log^2 (MN))$ with probability at least $1-\frac{2}{(MN)^2}$.
\end{corollary}

\section{Adaptive Algorithm}
\label{section:adaptivealgorithm}

A limitation of Algorithms \ref{algorithm:offline-greedy} and \ref{algorithm:online-greedy}
is that $C$ needs to be known ahead for each window $W$ that the algorithms are applied to.
We show here that it is possible to guess the value $C$ in a window $W$.
We present the Algorithm {\sf Adaptive-Greedy} (Algorithm \ref{algorithm:adaptive-greedy})
which can guess the value of $C$.
From the analysis of Algorithms \ref{algorithm:offline-greedy} and \ref{algorithm:online-greedy},
we know that the knowledge of the value $C$ plays vital role in the probability of success of the algorithms.

\begin{algorithm}[t]
{\small
\KwIn{An $M \times N$ execution window $W$ with $M$ threads each with $N$ transactions, where $C$ is unknown\;}
\KwOut{A greedy execution schedule for the window of transactions\;}
\BlankLine

Associate triplet of priorities $\langle \pi^{(3)}, \pi^{(2)}, \pi^{(1)} \rangle$ 
to each transaction when available for execution\;
\textbf{Code for thread $P_i$}\;
\Begin{
Initial contention estimate $C_i \leftarrow 1$\;
\Repeat {all transactions are committed}
{
{\sf Online-Greedy($C_i$, $W$)}\;
\If {bad event}{$C_i\leftarrow 2 \cdot C_i$ ;
}
}
\caption{{\sf Adaptive-Greedy}}
\label{algorithm:adaptive-greedy}
}
}
\end{algorithm}

In {\sf Adaptive-Greedy} each thread $P_i$ attempts to guess individually 
the right value of $C$.
The algorithm works based
on the exponential back-off strategy used by many contention managers
developed in the literature such as {\sf Polka}. 
The algorithm works as follows: each thread starts with assuming $C = 1$.
Based on the current estimate $C$ then the thread attempts to execute Algorithm \ref{algorithm:online-greedy},
for each of its transactions assuming the window size $M \times N$.
Now, if the choice of $C$ is correct then each transactions of the thread in the window $W$ 
of the thread $P_i$ should commit within the designated frame that it becomes high priority.
Thus, all transactions of the frame
should commit within the makespan time estimate Algorithm \ref{algorithm:online-greedy}
which is $\tau_C = O(C \log (MN) + N \log^2(MN))$.
However, if during $\tau_C$ some thread does not commit within its designated frame (bad event),
then thread $P_i$ will assume that the choice of $C$ was incorrect,
and will start over again with the remaining transactions assuming $C = 2C'$,
where $C'$ is the previous estimate for $C$.
Eventually thread $P_i$ will guess the correct value of $C$ for the window $W$,
and all its transactions will commit within the respective time.  
 
The different threads adapt independently from each other to the correct value of $C$.
At the same moment of time the various threads may have assumed different values of $C$.
The threads with higher estimate of $C$ will be given higher priority
in conflicts, since threads with lower $C$ most likely have guessed the wrong $C$ 
and are still adapting.
In order to handle conflicts each transaction uses a vector of priorities 
with three values $\langle \pi^{(3)}, \pi^{(2)}, \pi^{(1)} \rangle$.
The value of priority entry $\pi^3$ is inversely proportional to the current guess of $C$ for the thread,
so that higher value of $C$ implies higher priority.
The last two entries $\pi^{(2)}$ and $\pi^{(1)}$ are the same as in Algorithm \ref{algorithm:online-greedy}.
It is easy to that the correct choice of $C$ will be reached by a thread $P_i$ within $\log C$ iterations.
The total makespan and response time is asymptotically the same as with Algorithm \ref{algorithm:online-greedy}.

\section{Optimal Window Decomposition}
\label{section:decomposition}

In this section we are interested in partitioning a $M \times N$ window $W$ into some decomposition of
sub-windows such that if we schedule the transactions of each sub-window separately
using one of our greedy contention managers
then the sum of the makespans of the sub-windows is better than scheduling all the transactions
of $W$ as a single window.
In particular we are seeking a decomposition that minimizes the maximum {\em density} of the sub-windows,
where the density expresses how much larger is the contention with respect to the number of transactions per thread.

For window $W$ with congestion $C$ we define the density as $r = C/N$.
Consider some decomposition $D$ of window $W$ into
different sub-windows $D = \{W_1, \cdots, W_k\}$,
where sub-window $W_i$ has respective size $M \times X_i$.
Let $C_i$ denote the contention of window $w_i$.
The density of $W_i$ is $r_i = C_i / X_i$.
Let $r_D = \max_{W_i \in D} r_i$.
The {\em optimal window decomposition} $D^*$
has density $r_{D^*} = min_{D \in {\cal D}} r_D$,
where ${\cal D}$ denotes the set all possible decompositions of $W$.
Note that different decompositions in ${\cal D}$ may have different number of windows.
Two example decompositions members of ${\cal D}$ is one that consists only of $W$,
and another that consists of all single column windows of $W$.

The optimal window decomposition $D^*$ can provide asymptotically better makespan for $W$ if $r_{D^*} = o(r)$.
Using one of our greedy algorithms, the makespan of each sub-window $W_i \in D^*$
is ${\widetilde O}((1+r_{D^*})X_i)$ (where the notation ${\widetilde O}$ hides polylog factors).
Thus, using $D^*$, the makespan for the whole window $W$ becomes
${\widetilde O} ((1+r_{D^*})\sum_{W_i \in D^*} X_i) = {\widetilde O} ((1+r_{D^*}) N)$.
If we apply one of our greedy algorithms in the whole window $W$ directly,
then the makespan for $W$ is ${\widetilde O} ((1+r) N)$,
which may be asymptotically worse than using the optimal decomposition $D^*$ when $r_{D^*} = o(r)$.

\begin{figure}[ht]
\centering
\includegraphics[height=2.1in,width=4.5in]{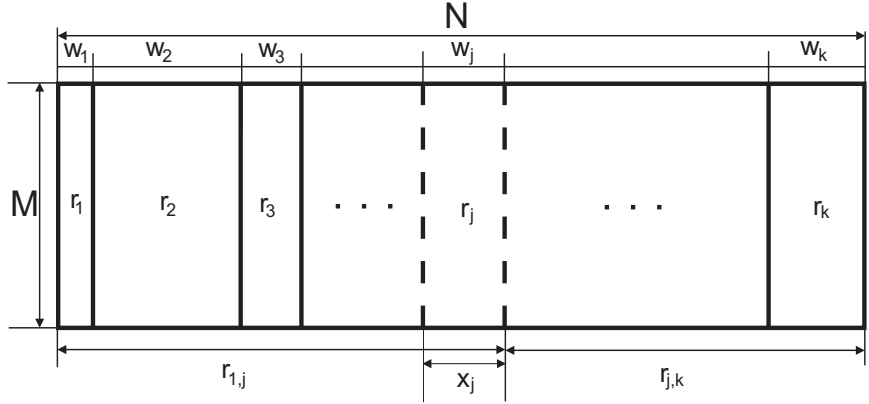}
\caption{Optimal window decomposition}
\label{figure:optimalwindow}
\end{figure}

We use a dynamic programming approach to compute the optimal decomposition $D^*$ of $W$.
The idea is compute the optimal decomposition of all prefix windows of $W$.
As shown in Figure \ref{figure:optimalwindow},
our goal is to determine the optimal window decomposition
including the prefix window up to column $k$
provided that optimal window decomposition till column $k-1$ has been already computed.
In this case, there are $k$ possible combinations to examine for finding the optimal window size which will
minimize the maximum of all the contention densities.
The details are in the proof of the following theorem.

\begin{theorem}[optimal window decomposition]
\label{lemma: window-decomposition}
The optimal window decomposition $D^*$ for an arbitrary $M \times N$ window $W$
can be computed in polynomial time.
\end{theorem}

\begin{proof}
From the problem description, we can readily see the overlapping-subproblems property in the optimal window decomposition problem.
Let $r_{j,k}$ denote the density in the decomposition of the sub-window $W_{j,k}$,
which starts at column $j$ and ends at column $k$, where $j \leq k$.
Let $r^*_{j,k}$ denote the maximum density in the optimal decomposition of the sub-window $W_{j,k}$.
The optimal window decomposition in this scenario
can be determined from this recursive formula:
$$r^*_{j,k}=\displaystyle{\mathop{\mbox{min}}_{1\leq j\leq k-1}}\{\max(r^*_{1,j}, (r_{j,k}))\}.$$
To find the optimal window decomposition for the $k$-th prefix window $W_{1,j}$,
we have to check for all the combinations from first to $k-1$ prefix window and the suffix up to $k$.
Using the formula we can compute $r^*_{1,k}$ for each prefix $W_{1,k}$.
Our algorithm needs $O(k)$ time to compute optimal window size for the $k$-th prefix
provided that the optimal window computation till the $(k-1)$-th prefix is known.
To compute then all the values for each window combination from 1 to $k$,
our algorithm recursively takes $O(k^2)$ steps.
The final density is $r_{D^*} = r^*_{1,N}$.
\end{proof}

\section{Conclusions} \label{section:conclusion}
In this paper, we consider greedy contention managers for transactional memory
for $M \times N$ windows of transactions with $M$ threads and $N$ transactions per thread
and present three new algorithms for contention management in transactional memory from a worst-case perspective.
These algorithms are efficient, adaptive, and handle windows of transactions
and improve on the worst-case performance of previous results.
These are the first such results for the execution of sequences of transactions instead of the one-shot problem used in other literature.
our algorithms present new trade-offs in the analysis of greedy contention managers
for transactional memory.
We also show that the optimal window decomposition can be determined using dynamic programming for any arbitrary window.
With this work, we are left with some issues for future work.
One may consider arbitrary time durations for the transactions to execute instead of the $O(1)$ time we considered in our analysis.
We believe that our results scale by a factor proportional to the longest transaction duration.
The other aspects may be to explore in deep the alternative algorithms where the randomization does not occur
at the beginning of each window but rather during the executions of the algorithm by inserting random periods of low priority
between the transactions in each thread. One may also consider the dynamic expansion and contraction of the execution window to preserve the congestion measure $C$. Thus, the execution window will not be a part of the algorithm but only a part of the analysis.
This will result to more practical algorithms which at the same time
achieve good performance guarantees.

\small
\bibliographystyle{plain}
\bibliography{podc2010}

\begin{thebibliography}{10}

\bibitem{Ansari09}
Mohammad Ansari, Christos Kotselidis, Mikel Lujan, Chris Kirkham, and Ian
  Watson.
\newblock On the performance of contention managers for complex transactional
  memory benchmarks.
\newblock In {\em In Proceedings of the 8th International Symposium on Parallel
  and Distributed Computing (ISPDC'09)}, July 2009.

\bibitem{Att06}
Hagit Attiya, Leah Epstein, Hadas Shachnai, and Tami Tamir.
\newblock Transactional contention management as a non-clairvoyant scheduling
  problem.
\newblock In {\em PODC '06: Proceedings of the twenty-fifth annual ACM
  symposium on Principles of distributed computing}, pages 308--315, New York,
  NY, USA, 2006. ACM.

\bibitem{Gue05b}
R.~Guerraoui, M.~Herlihy, M.~Kapalka, and B.~Pochon.
\newblock Robust {C}ontention {M}anagement in {S}oftware {T}ransactional
  {M}emory.
\newblock In {\em Proceedings of the {OOPSLA} 2005 {W}orkshop on
  {S}ynchronization and {C}oncurrency in {O}bject-{O}riented {L}anguages
  ({SCOOL}'05)}, 2005.

\bibitem{Gue05a}
Rachid Guerraoui, Maurice Herlihy, and Sebastian Pochon.
\newblock Toward a theory of transactional contention management.
\newblock In {\em Proceedings of the Twenty-Fourth Annual Symposium on
  Principles of Distributed Computing (PODC)}, 2005.

\bibitem{Har03}
Tim Harris and Keir Fraser.
\newblock Language support for lightweight transactions.
\newblock In {\em Object-Oriented Programming, Systems, Languages, and
  Applications}, pages 388--402. Oct 2003.

\bibitem{Har05}
Tim Harris, Simon Marlow, Simon Peyton-Jones, and Maurice Herlihy.
\newblock Composable memory transactions.
\newblock In {\em PPoPP '05: Proceedings of the tenth ACM SIGPLAN symposium on
  Principles and practice of parallel programming}, pages 48--60, New York, NY,
  USA, 2005. ACM.

\bibitem{Her03}
Maurice Herlihy, Victor Luchangco, Mark Moir, and William~N. Scherer, III.
\newblock Software transactional memory for dynamic-sized data structures.
\newblock In {\em PODC '03: Proceedings of the twenty-second annual symposium
  on Principles of distributed computing}, pages 92--101, New York, NY, USA,
  2003. ACM.

\bibitem{Herlihy93transactionalmemory}
Maurice Herlihy and J.~Eliot~B. Moss.
\newblock Transactional memory: Architectural support for lock-free data
  structures.
\newblock In {\em in Proceedings of the 20th Annual International Symposium on
  Computer Architecture}, pages 289--300, 1993.

\bibitem{LMR94}
F.~T. Leighton, B.~M. Maggs, and S.~B. Rao.
\newblock Packet routing and job-shop scheduling in ${O}(congestion+dilation)$
  steps.
\newblock {\em Combinatorica}, 14:167--186, 1994.

\bibitem{Ramadan08}
Hany~E. Ramadan, Christopher~J. Rossbach, Donald~E. Porter, Owen~S. Hofmann,
  Aditya Bhandari, and Emmett Witchel.
\newblock Metatm/txlinux: Transactional memory for an operating system.
\newblock {\em IEEE Micro}, 28(1):42--51, 2008.

\bibitem{Sch05}
William~N. Scherer, III and Michael~L. Scott.
\newblock Advanced contention management for dynamic software transactional
  memory.
\newblock In {\em PODC '05: Proceedings of the twenty-fourth annual ACM
  symposium on Principles of distributed computing}, pages 240--248, New York,
  NY, USA, 2005. ACM.

\bibitem{Sch04}
William~N. {{Scherer III}} and Michael~L. Scott.
\newblock Contention management in dynamic software transactional memory.
\newblock In {\em Proceedings of the ACM PODC Workshop on Concurrency and
  Synchronization in Java Programs}, St. John's, NL, Canada, Jul 2004.

\bibitem{scherer05}
William~N. {{Scherer III}} and Michael~L. Scott.
\newblock Randomization in \textsc{STM} contention management
  (\textsc{POSTER}).
\newblock In {\em Proceedings of the 24th ACM Symposium on Principles of
  Distributed Computing}, Las Vegas, NV, Jul 2005.
\newblock Winner, most popular poster presentation award.

\bibitem{Schneider09}
Johannes Schneider and Roger Wattenhofer.
\newblock Bounds on contention management algorithms.
\newblock In {\em The proceedings of the 20th International Symposium on
  Algorithms and Computation (ISAAC 2009)}, December 16-18, 2009.

\bibitem{Shavit95}
Nir Shavit and Dan Touitou.
\newblock Software transactional memory.
\newblock In {\em PODC '95: Proceedings of the fourteenth annual ACM symposium
  on Principles of distributed computing}, pages 204--213, New York, NY, USA,
  1995. ACM.

\end{thebibliography}
\end{document}